\documentclass[12pt]{iopart}
\usepackage{iopams,graphicx,texdraw,epsf,bm}

\newcommand{\eqref}[1]{(\ref{#1})}
\newtheorem{theorem}{Theorem}[section]
\newtheorem{lemma}[theorem]{Lemma}

\newenvironment{proof}[1][Proof]{\begin{trivlist}
\item[\hskip \labelsep {\bfseries #1}]}{\end{trivlist}}

\newenvironment{remark}[1][Remark]{\begin{trivlist}
\item[\hskip \labelsep {\bfseries #1}]}{\end{trivlist}}
\eqnobysec
\begin{document}
\paper[Spectrum of discrete Schr\"odinger operator with even potential]{On the spectrum of the discrete $1d$ Schr\"odinger  
operator with an arbitrary even potential}
\author{ S.~B.~Rutkevich}  
\address{Fakult\"at f\"ur Physik, Universit\"at Duisburg-Essen, D-47058 Duisburg, Germany}
\ead{sergei.rutkevich@uni-due.de}
\begin{abstract}{The  discrete one-dimensional Schr\"odinger operator 
 is studied  in the finite interval of length $N=2 M$ with the Dirichlet boundary conditions and an arbitrary 
potential even with respect to the spacial reflections. It is shown, that the eigenvalues of such a discrete 
Schr\"odinger operator (Hamiltonian), which is represented by the  $2M\times2M$ tridiagonal matrix, satisfy  a set of
polynomial constrains. 
The most interesting constrain, which is explicitly obtained, leads to the effective Coulomb interaction
between the Hamiltonian  eigenvalues. In the limit $M\to\infty$, this constrain induces the requirement, which should  
satisfy the scattering date in the scattering problem for the discrete Schr\"odinger  
operator in the half-line. We obtain such a requirement in the simplest case of the Schr\"odinger operator,
which does not have bound and semi-bound states, and which 
potential has a compact support.
}
\end{abstract}
\section{Introduction}
Consider the discrete Schr\"odinger eigenvalue problem in the one-dimensional chain having even number of sites $N=2M$, 
with  an arbitrary  real even potential \newline$V=\{v_j\}_{j=1}^N$:
\begin{eqnarray} \label{eig}
&&v_j\,\psi_l(j)+\left[2\psi_l(j)-\psi_l(j-1)-\psi_l(j+1)\right]=\lambda_l \psi_l(j), \\
&&v_{N+1-j}=v_j, \label{ev}\\
&&j=1,\ldots,N, \quad \quad
\lambda_1<\lambda_2<\ldots<\lambda_N .\nonumber
\end{eqnarray}
Eigenstates of  (\ref{eig}) are
subjected to the Dirichlet boundary conditions 
\begin{equation}\label{BC}
\psi_l(0)= \psi_l(N+1)=0.
\end{equation}

The discrete Sturm-Liouvelle problem \eqref{eig}-\eqref{BC} without the parity constrain \eqref{ev} plays an important role
in the theory of Anderson localization \cite{car1990,past2011}. The problem \eqref{eig}-\eqref{BC} with an even potential 
\eqref{ev} naturally arrises
in the context of the theory of the thermodynamic Casimir effect \cite{DGHHRSxxx12}. 

It is proved in this paper, that the eigenvalues of the discrete Sturm-Liouville problem \eqref{eig}-\eqref{BC} satisfy the 
following equality:
\begin{equation}\label{pr0}
\prod_{m=1}^M\prod_{n=1}^M(\lambda_{2m-1}-\lambda_{2n})=2^M\,(-1)^{M(M+1)/2}.
\end{equation}

One can easily check, that \eqref{pr0} is satisfied for  small $M=1,2,\ldots$ For arbitrary natural $M$, relation \eqref{pr0} 
is proved in Section \ref{SLP}. Section \ref{SchPr} contains some well-known basic facts about the 
scattering problem in the half-line for the discrete Schr\"odinger operator. 
In the limit $M\to\infty$, equality \eqref{pr0} leads to certain constrains on the scattering data in such a problem, which are derived in Section \ref{infM}.
Concluding remarks are given in Section \ref{conc}. 
Proof of \eqref{pr0} for the free case $v_j=0$, $j=1,\ldots,M$ is presented in \ref{free}. 
\section{Discrete Sturm-Liouville problem in the finite interval   \label{SLP}}
In the case of zero potential $v_j=0$, the
solution of (\ref{eig})-\eqref{BC}  reads as
\begin{eqnarray}
\psi_l(j)=\sin(k_l j),\\\label{fr-eig}
\lambda_l= \omega(k_l),\\
k_l=\frac{\pi l}{N+1}, \label{kl}
\end{eqnarray}
with  
\begin{equation}\label{spe}
\omega(p)=4\sin^2(p/2), 
\end{equation}
and $l=1,\ldots,N$.

For a general real even potential $v_j$, the eigenstates
$\psi_{2m-1}(j)$, $m=1,\ldots,M$  are even with respect to the reflection
\begin{equation}
\psi_{2m-1}(N+1-j)=\psi_{2m-1}(j),
\end{equation}
whereas eigenstates $\psi_{2m}(j)$, $m=1,\ldots,M$  are odd:
\begin{equation}
\psi_{2m}(N+1-j)=-\psi_{2m}(j).
\end{equation}

It is useful to consider two associated eigenvalue problems for the even and odd states, which
are restricted to the half-chain $j=1,\ldots,M$.
The  eigenvectors $\psi_{2m-1}(j)$, $j=1,\ldots, M$,  are the eigenstates
of the tridiagonal $M\times M$ matrix $H^{(ev)}$:
\begin{eqnarray}\label{Hev}
H^{(ev)}=\left(\begin{array}{ccccccc}
b_1 &-1&0&0&0&\dots&0\\
-1&b_2&-1&0&0&\dots&0\\
0&-1&b_3&-1&0&\dots&0\\
\dots&\dots&\dots&\dots&\dots&\dots&\dots\\
0&0&\dots&0&-1&b_{M-1}&-1\\
0&0&\dots&0&0&-1&b_M-1
\end{array}\right),\\
b_j= v_j+2, \nonumber
\end{eqnarray}
with eigenvalues $\mu_m=\lambda_{2m-1}$, $m=1,\ldots,M$.
Similarly, the
eigenvectors $\psi_{2m}(j)$, $j=1,\ldots, M$,  are the eigenstates
of the tridiagonal $M\times M$ matrix $H^{(od)}$:
\begin{equation}\label{Hod}
H^{(od)}=\left(\begin{array}{ccccccc}
b_1 &-1&0&0&0&\dots&0\\
-1&b_2&-1&0&0&\dots&0\\
0&-1&b_3&-1&0&\dots&0\\
\dots&\dots&\dots&\dots&\dots&\dots&\dots\\
0&0&\dots&0&-1&b_{M-1}&-1\\
0&0&\dots&0&0&-1&b_M+1
\end{array}\right),
\end{equation}
with eigenvalues $\nu_m =\lambda_{2m}$, $m=1,\ldots,M$.
Note, that the matrices $H^{(ev)}$ and $H^{(od)}$ are simply related
\begin{equation} \label{P}
H^{(od)}-H^{(ev)}=2 P,
\end{equation}
with the projecting  matrix  $P_{m,m'}=\delta_{m,M}\delta_{m',M}$,
$\;\;m,m'=1,\ldots,M$, and ${\rm rank}\,P=1$.

It is convenient to allow the potential  $\{b_j\}_{j=1}^M$ in the diagonal of the matrices \eqref{Hev},  \eqref{Hod}  to take complex values.
\begin{lemma} The matrices $H^{(od)}$ and $H^{(ev)}$ defined by \eqref{Hev}, \eqref{Hod} \label{Le}
have no common eigenvalues for arbitrary complex
potential $\{b_j\}_{j=1}^M$.
\end{lemma}
\begin{proof}
We will assume  that the matrices $H^{(ev)}$ and $H^{(od)}$ have a common
eigenvalue $\Lambda$ and come to  contradiction.

So, let us suppose that
\[
\sum_{j'=1}^M H^{(ev)}_{j,j'}x_{j'}=\Lambda\, x_j, \quad
\sum_{j'=1}^M H^{(od)}_{j,j'}y_{j'}=\Lambda \,y_j,
\]
with nonzero vectors $\{x_j\}_{j=1}^M$, $\{y_j\}_{j=1}^M$.
We get
\begin{eqnarray*}
\fl \Lambda \sum_{j=1}^M y_j\,x_j=\sum_{j=1}^M \sum_{j'=1}^M y_jH^{(ev)}_{j,j'}x_{j'}=
\sum_{j=1}^M \sum_{j'=1}^M y_j(H^{(od)}_{j,j'}-2P_{j,j'})x_{j'}=\\
\fl\sum_{j=1}^M \sum_{j'=1}^M y_jH^{(od)}_{j,j'}x_{j'}
-2\sum_{j=1}^M \sum_{j'=1}^M y_jP_{j,j'}x_{j'}=-2 y_M\,x_M+\Lambda \sum_{j=1}^M y_j\,x_j.
\end{eqnarray*}
Here we have taken into account, that the matrix $H^{(od)}$ is symmetric.
Thus,
\[
y_M\,x_M=0,
\]
which means that at least one of the numbers $y_M$ and $x_M$ is zero.
However, if $y_M=0$, we conclude immediately\begin{footnote} 
{Really, if $0=y_M\equiv \psi(j=M)$, then $\psi(j=M+1)=-\psi(j=M)=0$,
since $\psi(j)=-\psi(2M+1-j)$ for all $j=1,\dots,2M$.
And since the wave-function $\psi(j)$ takes zero values at two neighbor
sites $\psi(M)=\psi(M+1)=0$, one can check recursively from (\ref{eig}),
that $\psi(M-1)=0, \;\psi(M-2)=0, \;\dots,\; \psi(1)=0$,
and, therefore, $\psi(j)=0$ for all $j=1,\dots,2M$.}
\end{footnote}, that $y_m=0$ for all $m=1,\ldots,M$,
providing that $\Lambda$ {\it is not} an eigenvalue of $H^{(od)}$.
Similarly, if $x_M=0$, we conclude, that $x_m=0$ for all $m=1,\ldots,M$,
providing that $\Lambda$ {\it is not} an eigenvalue of $H^{(ev)}$.
This contradiction with the initial assumption proofs the Lemma.
\end{proof}

The sets of eigenvalues $\{\mu_m\}_{m=1}^M$ and $\{\nu_m\}_{m=1}^M$
of the matrices  $H^{(ev)}$ and  $H^{(od)}$ are not independent, but are
subjected to certain polynomial constrains following from (\ref{P}):
\begin{eqnarray} \label{c}
\fl\sum_{m=1}^M \mu_m={\rm Tr}H^{(ev)}=-2+{\rm Tr}H^{(od)}=-2+\sum_{m=1}^M \nu_m,\\\nonumber
\fl\sum_{m=1}^M \mu_m^2={\rm Tr}(H^{(ev)})^2={\rm Tr}(-2P+H^{(od)})^2=-4b_M+\sum_{m=1}^M \nu_m^2,\\\nonumber
\fl\sum_{m=1}^M \mu_m^3={\rm Tr}(H^{(ev)})^3={\rm Tr}(-2P+H^{(od)})^3=-8-6b_M^2+\sum_{m=1}^M \nu_m^3,\\\nonumber
\fl\sum_{m=1}^M \mu_m^4={\rm Tr}(H^{(ev)})^4={\rm Tr}(-2P+H^{(od)})^4
=-8b_{M-1}-24b_M-8b_M^2+\sum_{m=1}^M \nu_m^4,\\\nonumber
\fl\sum_{m=1}^M \mu_m^5={\rm Tr}(H^{(ev)})^5={\rm Tr}(-2P+H^{(od)})^5=\\ \nonumber
\fl -32-10b_{M-1}^2-
20 b_{M-1}b_M-50 b_M^2-10b_M^4+\sum_{m=1}^M \nu_m^5,\\\nonumber
\fl\ldots\ldots\ldots\ldots\ldots\ldots\ldots\ldots\ldots\ldots\ldots\ldots\ldots
\ldots\ldots\ldots\ldots\ldots\ldots\ldots.
\end{eqnarray}
\begin{remark}
\item{(1)} Equations (\ref{c}) provide a simple way to solve the inverse
spectral problem, i.e. to determine one by one the
potential $b_M, \, b_{M-1}, \ldots,b_1$,  if the
both sets of eigenvalues
$\{\mu_m\}_{m=1}^M$ and $\{\nu_m\}_{m=1}^M$ are known.
\item{(2)} Excluding one by one the
potential $b_M, \, b_{M-1}, \ldots,b_1$ from equations \eqref{c}, one can obtain the infinite set 
of polynomial constrains of increasing degrees on the eigenvalues $\{\mu_m\}_{m=1}^M$, $\{\nu_m\}_{m=1}^M$.
No more than $M$ of constrains in this set can be independent, since the above mentioned eigenvalues 
are determined 
by $M$ parameters 
$\{v_j\}_{j=1}^M$ as zeroes of  the characteristic polynomials of the matrices \eqref{Hev} and \eqref{Hod}. 
\item{(3)} One can easily see from (\ref{c}), that
the symmetric polynomials of eigenvalues $\{\mu_m\}_{m=1}^M$,
as well as the symmetric polynomials of eigenvalues $\{\nu_m\}_{m=1}^M$,
can be written as polynomial functions of the potential $\{b_j\}_{j=1}^M$.
\end{remark}
Now we are ready to prove the main
\begin{theorem} For arbitrary complex numbers $\{b_j\}_{j=1}^M$, the eigenvalues $\{\mu_m\}_{m=1}^M$ and  
$\{\nu_m\}_{m=1}^M$ of the 
matrices $H^{(od)}$ and $H^{(ev)}$ defined by \eqref{Hev}, \eqref{Hod} satisfy the equality:
\begin{equation}\label{pr}
\prod_{m=1}^M\prod_{n=1}^M(\mu_m-\nu_{n})=2^M\,(-1)^{M(M+1)/2}.
\end{equation}
\end{theorem}
By restriction of this result to real  $\{b_j\}_{j=1}^M$, we arrive to \eqref{pr0}.
\begin{proof}
The left-hand side of (\ref{pr}) is a symmetric polynomial function of $\{\mu_m\}_{m=1}^M$.
It is also a a symmetric polynomial function of $\{\nu_m\}_{m=1}^M$. Due to Remark (3),
we can conclude, that the left-hand side of (\ref{pr}) can be written as a polynomial function $Q_M(b)$ of
the potential $\{b_j\}_{j=1}^M$:
\begin{equation}\label{pr1}
\prod_{m=1}^M\prod_{n=1}^M(\mu_m-\nu_{n})=Q_M(b).
\end{equation}
It follows from this relation and  Lemma \ref{Le},
that the polynomial  function $Q_M(b)$ of $M$ complex variables
$\{b_j\}_{j=1}^M$ has no zeros. This means, that this  function is a
constant, $Q_M(b)\equiv C_M$, which {\underline {must not}} depend on the potential  $\{b_j\}_{j=1}^M$.

One can now determine this constant $C_M$ using an appropriate convenient  choice of the
potential. To this end, we  put
\begin{equation}
b_j=2, \quad j=1,\ldots,M,
\end{equation}
which corresponds to $v_j=0$, $j=1,\ldots,M$.
Reminding  (\ref{fr-eig}), we get
\begin{equation}
\mu_m= 4 \sin^2\frac{ (2m-1)\pi}{2(2M+1)} , \quad
\nu_m= 4 \sin^2\frac{2 m \pi}{2(2M+1)},
\end{equation}
and equality we need to prove for all natural $M$  takes the form
\begin{equation}\label{pr2}
\fl\prod_{m=1}^M\prod_{n=1}^M\left[4 \sin^2\frac{ (2m-1)\pi}{2(2M+1)}-
 4 \sin^2\frac{2 n \pi}{2(2M+1)}\right]=2^M\,(-1)^{M(M+1)/2}.
\end{equation}
Proof of this formula is given in \ref{free}.
\end{proof}

It is interesting to note, that equality \eqref{pr0} allows the electrostatic interpretation. 
Really, let us take the logarithm of the absolute values of the both sides of \eqref{pr0}, 
and rewrite the result in the form
\begin{equation}\label{UCoul}
-\sum_{m=1}^M\sum_{n=1}^M\ln|x_m^{(A)}-x_n^{(B)}|=-M\, \ln 2,
\end{equation}
where $x_m^{(A)}=\lambda_{2m-1}$, and $x_n^{(B)}=\lambda_{2n}$, with $m,n=1,\ldots,M$ will be 
treated as space coordinates of two different sets  of $M$ particles of types $A$ and $B$,
which are 
distributed along the  $x$-axis in the two-dimensional plane. Particles of the  $A$ type interlace with particles
of the $B$ type, $x_m^{(A)}<x_m^{(B)}<x_{m+1}^{(A)}.$ 
If particles of the same type do not interact with each other,
and particles of different types interact via the pair 
$2d$ Coulomb potential $u(x^{(A)},x^{(B)})=-\ln|x^{(A)}-x^{(B)}|$, then
equation \eqref{UCoul} states simply, that the total Coulomb energy of this system of  $2M$ particles 
should be equal to
$-M\ln 2$.
\section{Scattering problem for the discrete Schr\"odinger operator in the half-line\label {SchPr}}
In this Section we briefly summarize some well-known basic results from the
scattering theory in the half-line (see, for example \cite{ChaSab,Ca73}) adapted for the the 
discrete Schr\"odinger operator \cite{car1990,past2011}.

Consider the discrete Schr\"odinger equation (\ref{eig}) in the half-line $j\in \mathbb{N}$ 
\begin{eqnarray}\label{Sch}
 \left(H \psi\right)_j=\lambda\, \psi(j),\\
  \left(H \psi\right)_j=    \label{H}
 v_j\,\psi(j)+\left[2\psi(j)-\psi(j-1)-\psi(j+1)\right], \\
 j=1,2, \ldots,\infty,  \nonumber
\end{eqnarray}
supplemented with the Dirichlet boundary condition 
\begin{equation}\label{Dbc}
\psi(0)=0.
\end{equation}
The potential $V=\{v_j\}_{j=1}^\infty$ in \eqref{Sch} is the infinite sequence of real numbers. In the scattering theory, 
the potential should vanish fast enough at infinity. It is usually required \cite{ChaSab}, that 
\begin{equation}
\sum_{j=1}^\infty j |v_j|<\infty.\label{Vinf}
\end{equation}
For such a potential, the  spectrum $\sigma[H]$ of the operator  $H$ defined by 
\eqref{Sch}-\eqref{Dbc} consists of the continuous part $\sigma_{cont}[H]=(0,4)$ 
and a finite number of  discrete eigenvalues.

At a given $\lambda$, equations \eqref{Sch}, \eqref{H}  with omitted boundary condition \eqref{Dbc}
 have two linearly  independent solutions, and the general 
solution of  \eqref{Sch}, \eqref{H}  can be written as their linear combination. For two sequences $\{\psi_1(j)\}_{j=0}^\infty$, and 
$\{\psi_2(j)\}_{j=0}^\infty$, one can define the Wronskian 
\begin{equation}
W[\psi_1,\psi_2]_j=\psi_1(j)\psi_2(j+1)-\psi_1(j+1)\psi_2(j), \quad j=0,1,2,\ldots
\end{equation}
It is straightforward to check, that the Wronskian of two solutions of equations \eqref{Sch}, \eqref{H} does not depend on $j$.

Let us turn now to the scattering problem associated with equations\eqref{Sch}-\eqref{Dbc}, which 
corresponds to  the case $0<\lambda<4$.
Instead of parameter $\lambda\in \sigma_{cont}[H]$, it is  also convenient  to  use the  momentum $p$ and the related complex
parameter $z=\rm\rme^{\rmi p}$:
\[
\lambda=2-2 \cos p=2-z-z^{-1}.
\]
Three solutions of  \eqref{Sch}, \eqref{Dbc} are important for the scattering problem.
\begin{itemize}
\item
The {\it regular solution} $\varphi(j,p)$, which is fixed by the boundary condition
\begin{equation}\label{bc8}
\varphi(0,p)=0, \quad \varphi(1,p)=1.
\end{equation}
\item
Two {\it Jost solutions} $f(j,p)$, and $f(j,-p)$, which are determined by their behavior at large $j\to\infty$,
and describe the out- and in-waves, respectively,
\begin{equation}\label{Jas}
f(j,\pm p)\to  \exp(\pm \rmi   p j)=z^{\pm j}, \quad{\rm {at}}\quad  j\to\infty.
\end{equation}
\end{itemize}
The regular solution $\varphi(j,p)$ can be represented  as a linear combination of two Jost
solutions,
\begin{equation}\label{FJ}
\varphi(j,p)=\frac{\rmi }{2\sin p}\left[F(p)f(j,-p)-F(-p)f(j,p)\right], \quad 0<p<\pi.
\end{equation}
The complex coefficient $F(p)$ in the above equation is known as  the Jost function.
It is determined by \eqref{FJ} for 
 real momenta $p$ in the 
 interval $p\in (-\pi,\pi)$, where it satisfies the relation
 \begin{equation}
 F(-p)=[F(p)]^*,
 \end{equation}
and  can be written as
\begin{equation}
F(p)=\exp[{\sigma(p)- \rmi  \eta(p) }].
\end{equation}
At large $j\to\infty$, the regular solution behaves as
\[
\varphi(j,p)\to\frac{A(p)}{\sin p} \sin[p\, j+\eta(p)], \quad j\to +\infty,
\]
where  $A(p)=\exp [\sigma (p)]$ is the scattering amplitude, and $\eta(p)$ is the scattering phase.
The latter can be defined in such a way, that $\eta(-p)=-\eta(p)$ for $-\pi<p<\pi$.

The following exact representation holds for the Jost function $F(p)$ in terms of the regular solution $\varphi(j,p)$:
\begin{equation} \label{JF}
F(p)=1+\sum_{j=1}^\infty \rme^{\rmi  pj}v_j \varphi(j,p),
\end{equation}
cf. equation (1.4.4) in \cite{ChaSab} in the continuous case.

For $|z|=1$, denote by $\hat{F}(z)$ the Jost function $F(p)$ expressed in the complex parameter $z$: $F(p)=\hat{F}(z=\rme^{\rmi  p})$.
The function  $\hat{F}(z)$ can be analytically continued into the circle $|z|<1$, where it has finite number of 
zeros $\{a_n\}_{n=1}^\mathfrak{N}$. These zeroes determine the discrete spectrum $\{\lambda_n\}_{n=1}^\mathfrak{N}$  of the
problem \eqref{Sch}-\eqref{Dbc}:
\begin{equation}
\lambda_n=2-a_n-a_n^{-1}, \quad {\rm for  } \quad n=1,\ldots,\mathfrak{N}.
\end{equation}
Of course, these eigenvalues are real in the boundary problem with a real potential. 

To simplify further analysis, we shall consider 
in the sequel the potentials which satisfy  the following requirements:
\begin{enumerate}
\item \label{Vcs}
The potential $V$ should have a compact support, i.e.
\begin{equation}\label{vj0}
v_j=0, \quad {\rm for\; all}\quad j>J,
\end{equation}
with some natural $J$.
\item \label{Jzl1}
The corresponding Jost function $\hat{F}(z)$ should not have zeroes inside the circle $|z|<1$, i.e. 
$\mathfrak{N}=0$. In other words, 
the spectrum $\sigma[H]$ should be purely continuous. 
\item \label{Jzpm1}
The Jost function $\hat{F}(z)$ should take non-zero values at $z=\pm 1$: $\hat{F}(1)\ne 0$, and $\hat{F}(-1)\ne 0$.
\end{enumerate}
Conditions \eqref{Jzl1} and \eqref{Jzpm1} imply, that the operator $H$ does not have bound and semi-bound states 
\cite{ChaSab}, 
respectively.  

For the potential satisfying \eqref{vj0}, only $J$ initial terms survive  in the sum in the right-hand side of \eqref{JF}.
Since $ \varphi(j,p)$ is a polynomial of the spectral parameter $\lambda=2-z-z^{-1}$  of the order $j-1$, 
the Jost function \eqref{JF} expressed in the parameter $z$ is a polynomial of the degree $2J-1$:
\begin{equation}
\hat{F}(z)=1+\sum_{j=1}^{2J-1}c_j(V)\,z^j=\prod_{n=1}^{2J-1}\left[1-\frac{z}{a_n(V)}\right],
\end{equation}
where the coefficients $c_j(V)$ polynomially depend on the potential $v_j$, $j=1,\ldots,J$. 
Due to the constrains (\ref{Jzl1}),  (\ref{Jzpm1}), we get
\begin{equation}\label{a}
|a_n(V)|>1, \quad{\rm for\; all} \quad n=1,\ldots, J.
\end{equation}

Let us periodically continue the scattering phase $\eta(p)$ from the interval $(-\pi,\pi)$ to the whole real axis
$p\in \mathbb{R}$.
It follows from \eqref{a}, that for a potential satisfying (\ref{Vcs})-\eqref{Jzpm1}, the scattering phase 
 is an analytical odd $2\pi$-periodical function in the whole real axis:
 $\eta(p)\in C^\infty(\mathbb{R}/2\pi \mathbb{Z})$.
 
The notation $\delta(\lambda)$  will be used for the scattering phase $\eta(p)$ expressed in terms of the 
spectral parameter
$\lambda$: $\eta(p)=\delta(\lambda=2-2 \cos p)$, for $0\le\lambda\le4$, and $0\le p\le \pi$. Conditions 
\eqref{Jzl1}, \eqref{Jzpm1} guarantee, that 
\begin{equation}\label{del}
\delta(0)=\delta(4)=0.
\end{equation}
\section{Constrains on the scattering data in the discrete Schr\"odinger  scattering problem 
in the half-line \label{infM}}
It is shown in this Section, that the scattering phase $\delta(\lambda)$ in the boundary 
problem \eqref{Sch}-\eqref{Dbc}
for the discrete Schr\"odinger operator in the half-line 
with  an arbitrary potential $V$ obeying \eqref{Vcs}-\eqref{Jzpm1} should satisfy the constrain 
\begin{equation}\label{EQ}
\int_0^4 \rmd\lambda\,\delta(\lambda)\,
\frac{\lambda-2}{\lambda(\lambda-4)}+\frac{1}{\pi}\int_0^4 \rmd\lambda_1\, \delta(\lambda_1)\,\mathcal{P}\!\!
\,\int_0^4 \rmd\lambda_2\, \frac{\delta'(\lambda_2)}{\lambda_2-\lambda_1}=0,
\end{equation}
where $\mathcal{P}\!\!\int$ indicates the principal value integral. It is straightforward to rewrite the above constrain 
in the equivalent form in terms of the Jost function $\hat{F}(z)$:
\begin{equation}\label{Jost}
\fl{\ln \hat{F}(z=1)+\ln \hat{F}(z=-1)}+\oint_{|z|=1} \frac{\rmd z}{2\pi  \rmi }\, \ln [\hat{F}(1/z)]\frac{\rmd\ln[\hat{F}(z)]}{\rmd z}=0,
\end{equation}
where the integration path in the right-hand side is gone  in the counter-clockwise direction.

To prove \eqref{EQ}, let us consider the discrete Schr\"odinger  eigenvalue problem  \eqref{eig}-\eqref{BC}
in the finite interval $1\le j \le N=2M$, with $M>J$, and with the 
even potential $V^{(M)}=\{v_j^{(M)}\}_{j=1}^{2M}$, which restriction to the interval 
$[1,M]$ coincides with that of the potential $V=\{v_j\}_{j=1}^\infty$: 
\begin{equation}
v_j^{(M)}=\cases{
 v_j, \quad{\rm if} \quad j\le J,\\
  0, \quad{\rm if} \quad J<j\le 2M-J,\\
 v_{2M+1-j}, \quad{\rm if} \quad 2M-J< j \le 2M .
}
\end{equation}
It is easy to see, that the spectrum $\{\lambda_l\}_{l=1}^{2M}$ of the problem \eqref{eig}-\eqref{BC}
with such a potential can be expressed in terms of the scattering phase $\eta(p)$ of the 
corresponding  semi-infinite problem \eqref{Sch}-\eqref{Dbc} by the relations
\begin{eqnarray}
\lambda_l=\omega(p_l), \quad {l=1,\ldots,2M},\\
(2M+1)\,p_l +2\,\eta(p_l)=(2M+1)\,k_l, \label{pk}
\end{eqnarray}
where $\omega(p)=2-2 \cos p$, and $k_l=\pi l/(2M+1)$. Solving equation \eqref{pk} with respect to 
$p_l$ one obtains at large $M$:
\begin{equation}\label{pk3}
p_l=k_l-\frac{2\,\eta(k_l)}{2M+1}+\frac{4\, \eta(k_l)\,\eta'(k_l)}{(2M+1)^2}+O(M^{-3}).
\end{equation}

For an arbitrary $M>J$, we get from \eqref{pr0}:
\begin{equation}\label{Sm0}
\sum_{m=1}^M S_m=0,
\end{equation}
where
\begin{equation}\label{Sm}
S_m=\sum_{n=1}^M \left[
\ln|\omega(p_{2n-1})-\omega(p_{2m})|-\ln|\omega(k_{2n-1})-\omega(k_{2m})|
\right].
\end{equation}
Proceeding to the large-$M$ limit, one finds after substitution of  \eqref{pk3} into \eqref{Sm} and
expansion the result in $1/(2M+1)$:
\begin{equation}
S_m=S_m^{(0)}+ S_m^{(1)}+O(M^{-2}),
\end{equation}
where
\begin{eqnarray}\label{S0}
S_m^{(0)}=\frac{2}{2M+1}\sum_{n=1}^M \frac{\omega'(k_{2m})\,\eta(k_{2m})-\omega'(k_{2n-1})\,\eta(k_{2n-1})}
{\omega(k_{2n-1})-\omega(k_{2m})},\\\nonumber
S_m^{(1)}=\frac{2}{(2M+1)^2}\sum_{n=1}^M\Bigg\{
 \frac{\left[\omega'(k_{2n-1})\,\eta^2(k_{2n-1})\right]'-\left[\omega'(k_{2m})\,\eta^2(k_{2m})\right]'}
{\omega(k_{2n-1})-\omega(k_{2m})}\\
-\left[ \frac{\omega'(k_{2m})\,\eta(k_{2m})-\omega'(k_{2n-1})\,\eta(k_{2n-1})}
{\omega(k_{2n-1})-\omega(k_{2m})}
\right]^2\Bigg\}.\label{S1}
\end{eqnarray}
In the right-hand side of the second equation we can safely [up to the terms of order $O(M^{-2})$] replace the sum in $n$ by the 
integral over the momentum $q$:
\begin{eqnarray}\nonumber
S_m^{(1)}=\frac{1}{\pi(2M+1)}\int_0^\pi \rmd q\,\Bigg\{
 \frac{\left[\omega'(q)\,\eta^2(q)\right]'-\left[\omega'(k_{2m})\,\eta^2(k_{2m})\right]'}
{\omega(q)-\omega(k_{2m})}\\
-\left[ \frac{\omega'(k_{2m})\,\eta(k_{2m})-\omega'(q)\,\eta(q)}
{\omega(q)-\omega(k_{2m})}
\right]^2\Bigg\}+O(M^{-2}).\label{S1a}
\end{eqnarray}

Calculation of the large-$M$ asymptotics of $S_m^{(0)}$ is more delicate. 
First, we extend summation in \eqref{S0} in the index $n$ from 1 till $2M+1$
\begin{eqnarray}\nonumber
S_m^{(0)}=\frac{2}{2M+1}\sum_{n=1}^{M}R_m(k_{2n-1})=\\
\frac{2}{2M+1}\left[-
\frac{R_m(k_{2M+1})}{2}+\frac{1}{2}\sum_{n=1}^{2M+1} R_m(k_{2n-1})\right],\label{S0b}
\end{eqnarray}
where 
\begin{equation}\label{Rm}
R_m(q)=\frac{\omega'(k_{2m})\,\eta(k_{2m})-\omega'(q)\,\eta(q)}
{\omega(q)-\omega(k_{2m})}.
\end{equation}
In \eqref{S0b} we have taken into account the reflection symmetry  $R_m(q)=R_m(2\pi-q)$ 
of the function \eqref{Rm}, providing
$R_m(k_{2n-1})=R_m(k_{2(2M+1-n)+1})$.

Since $k_{2M+1}=\pi$, and $\eta(\pi)=0$, $\omega(\pi)=4$, we  get from \eqref{Rm}
\begin{equation}\label{Rm1}
R_m(k_{2M+1})=\frac{\omega'(k_{2m})\,\eta(k_{2m})}
{4-\omega(k_{2m})}.
\end{equation}
The sum  in the second line of \eqref{S0b} reads as 
\begin{equation}\label{Rm2}
\sum_{n=1}^{2M+1} R_m(k_{2n-1})=\sum_{n=1}^{2M+1} R_m\left(2\pi\,\frac{ n-1/2}{2M+1}\right).
\end{equation}
Since $R_m(q)\in C^\infty(\mathbb{R}/2\pi \mathbb{Z})$, 
this sum can be  replaced 
 with exponential accuracy by the integral at large $M\to\infty$:
\begin{equation}
\sum_{n=1}^{2M+1} R_m\left(2\pi\,\frac{ n-1/2}{2M+1}\right)=\frac{2M+1  }{2\pi}
\int_0^{2\pi}\rmd q \,R_m(q)+ o(M^{-\mu}),
\end{equation}
where $\mu$ is an arbitrary positive number, see formula 25.4.3 in \cite{AbrSt}.
Taking into account \eqref{Rm}, 
the integral in the right-hand side can be written as
\begin{eqnarray}\nonumber
\fl \int_0^{2\pi}\rmd q \,R_m(q)=\omega'(k_{2m})\,\eta(k_{2m}) \,\mathcal{P} \!\int_0^{2\pi}\frac{\rmd q }{\omega(q)-\omega(k_{2m})}
-\mathcal{P} \!\int_0^{2\pi} \rmd q\,\frac{\omega'(q)\,\eta(q)}{\omega(q)-\omega(k_{2m})}=\\ \label{Rint}
\fl 
2 \omega'(k_{2m})\,\eta(k_{2m})  \,\mathcal{P} \!\int_0^{\pi}
 \frac{\rmd q }{\omega(q)-\omega(k_{2m})}
-2\,I[\omega(k_{2m})]=-2\,I[\omega(k_{2m})],
\end{eqnarray}
where \begin{equation}\label{IL}
I(\Lambda)=\mathcal{P} \!\int_0^{4}
\rmd\lambda\,\frac{\delta(\lambda) }{\lambda-\Lambda}, \quad{\rm with}\quad 0<\Lambda<4.
\end{equation}
In the second line of \eqref{Rint} we have taken into account the equality 
\begin{eqnarray}\nonumber
 \mathcal{P}\!\int_0^\pi\, \frac{ \rmd q	}{[\omega(q)-\omega(k)]^\nu}\equiv \\
\fl\frac{1}{2}\lim_{\epsilon\to +0}\left\{
\int_0^\pi\, \frac{ \rmd q	}{[\omega(q+\rmi \epsilon)-\omega(k))]^\nu}+\int_0^\pi\, \frac{ \rmd q	}{[\omega(q-\rmi   \epsilon)-\omega(k)]^\nu}
\right\}=0,\label{eqP}
\end{eqnarray}
with $0<k<\pi$, and $\nu=1$.

Collecting \eqref{S0b}-\eqref{IL}, we get 
\begin{equation}\label{Sm0a}
S_m^{(0)}=-\frac{I[\omega(k_{2m})]}{\pi}-\frac{1}{2M+1}\frac{\omega'(k_{2m})\,\eta(k_{2m})}
{4-\omega(k_{2m})}+O(M^{-2}).
\end{equation}

Thus, we obtain from \eqref{Sm0a}  the equality
\begin{equation}\label{limS}
\lim_{M\to\infty}\sum_{m=1}^M S_m^{(0)}+\lim_{M\to\infty}\sum_{m=1}^MS_m^{(1)}=0,
\end{equation}
where $S_m^{(0)}$ and $S_m^{(1)}$ are given by equations \eqref{Sm0a}, and \eqref{S1}, respectively. 

In the second term, we can replace  with sufficient accuracy  the summation in $m$ by integration in the momentum $k$:
\begin{eqnarray}\nonumber
\sum_{m=1}^M S_m^{(1)}=\frac{1}{2\pi^2}\int_0^\pi \rmd k\int_0^\pi \rmd q\,
 \frac{\left[\omega'(q)\,\eta^2(q)\right]'-\left[\omega'(k)\,\eta^2(k)\right]'}
{\omega(q)-\omega(k)}\\
-\frac{1}{2\pi^2}\int_0^\pi \rmd k\int_0^\pi \rmd q\,\left[ \frac{\omega'(k)\,\eta(k)-\omega'(q)\,\eta(q)}
{\omega(q)-\omega(k)}
\right]^2+O(M^{-1}).\label{S1b}
\end{eqnarray}
The first integral in the right-hand side vanishes due to equality \eqref{eqP}
with $0<k<\pi$, and $\nu=1$. The second line in \eqref{S1b} can be transformed as follows
\begin{eqnarray}\label{aux}
-\frac{1}{2\pi^2}\int_0^\pi \rmd k\int_0^\pi \rmd q\,\left[ \frac{\omega'(k)\,\eta(k)-\omega'(q)\,\eta(q)}
{\omega(q)-\omega(k)}
\right]^2= \\
\fl -\frac{1}{\pi^2}\int_0^\pi \!\rmd k\, [\omega'(k)\,\eta(k)]^2\,\mathcal{P}\!\!\int_0^\pi  \,\frac{\rmd q }
{[\omega(q)-\omega(k)]^2}+
\frac{1}{\pi^2}\int_0^\pi \!\rmd k\,\mathcal{P}\!\!\int_0^\pi\!\rmd q \,\frac{  \omega'(k)\,\eta(k) \omega'(q)\,\eta(q)}
{[\omega(q)-\omega(k)]^2}.\nonumber
\end{eqnarray} 
The first term in the right-hand side vanishes due to equality \eqref{eqP} with $\nu=2$. 
Then, after a simple algebra we obtain from the second term in the right-hand side of \eqref{aux}
\begin{equation}\label{limS1}
\lim_{M\to\infty}\sum_{m=1}^MS_m^{(1)}=\frac{1}{\pi^2}\int_0^4 \rmd\lambda_1\, \delta(\lambda_1)\,\mathcal{P}\!\!
\,\int_0^4 \rmd\lambda_2\, \frac{\delta'(\lambda_2)}{\lambda_2-\lambda_1}.
\end{equation}

Let us turn  now to calculation of the first term in the left-hand side of equality \eqref{limS}. At large $M$, one obtains
\begin{eqnarray}\label{Sma}
\fl \sum_{m=1}^M S_m^{(0)} =-\frac{1}{2\pi}
\int_0^{\pi} \rmd k\,\frac{\omega'(k)\,\eta(k)}{4-\omega(k_{2m})}-\frac{1}{\pi}\sum_{m=1}^M I[\omega(k_{2m})]+O(M^{-1}).
\end{eqnarray}
The first term in the right-hand side equals to $I(4)/(2\pi)$.
The large-$M$ asymptotics of the sum in the right-hand side can be found as follows
\begin{eqnarray}\nonumber
\sum_{m=1}^M I[\omega(k_{2m})]=
-\frac{I(0)}{2}+\frac{1}{2}\sum_{m=1}^{2M+1} I[\omega(k_{2m})]=\\
 -\frac{I(0)}{2}+\frac{2M+1}{4\pi}\int_0^{2\pi} \rmd k\,I[\omega(k)]+O(M^{-\mu}),\label{S0d}
\end{eqnarray}
where $\mu$ is an arbitrary positive number. In deriving \eqref{S0d} we have taken 
into account, that $I[(\omega(k)]$ is the $2\pi$-periodical function of $p$ in $\mathbb{R}$, which is continuous 
with 
all its derivatives  [i.e.,  $I[(\omega(k)]\in C^\infty(\mathbb{R}/2\pi \mathbb{Z})$], and formula 25.4.3 in \cite{AbrSt}.
The integral in the right-hand side vanishes due to equality \eqref{eqP} with $\nu=1$:
\begin{equation}\label{Iint}
\fl\int_0^{2\pi} \rmd k\,I[\omega(k)]=\int_0^{2\pi} \rmd k\,\mathcal{P}\!\!\int_0^4 \rmd\lambda\, \frac{\delta(\lambda)}{\lambda-\omega(k)}=
\int_0^4 \rmd\lambda\, \delta(\lambda)\,\mathcal{P}\!\!\int_0^{2\pi} \frac{\rmd k}{\lambda-\omega(k)}=0.
\end{equation}
Collecting \eqref{Sma}-\eqref{Iint}, we come to the simple formula
\begin{equation}\label{limS0}
\lim_{M\to\infty}\sum_{m=1}^M S_m^{(0)} =\frac{I(0)+I(4)}{2\pi}.
\end{equation}
From \eqref{IL}, \eqref{limS}, \eqref{limS1}, \eqref{limS0}, we arrive to the final result \eqref{EQ}.

Similarly to \eqref{pr0}, equation \eqref{EQ} also admits the electrostatic interpretation.
Let us treat  the  function   $\rho(\lambda)=2\delta'(\lambda)/\pi$ as the electric charge density,
which is distributed in the linear interval $0<\lambda<4$. Requirement \eqref{del} implies, that
the total electric charge of this distribution is zero,
\[
\int_0^4 \rmd\lambda\, \rho(\lambda)=0.
\]
It is straightforward to rewrite \eqref{EQ} in terms of the function $\rho(\lambda)$:
\begin{eqnarray}\nonumber
-\frac{1}{2}\int_0^4 \rmd\lambda_1\,\rho(\lambda_1)\int_0^4 \rmd\lambda_2\,\rho(\lambda_2)\ln|\lambda_1-\lambda_2|-\\
\int_0^4 \rmd\lambda\,\rho(\lambda)\left[q_1\,
\ln \lambda+q_2\, \ln(4-\lambda)
\right]
-q_1 \,q_2 \ln 4=-\frac{\ln 2}{2},\label{Coul2}
\end{eqnarray}
where $q_1=q_2=-1/2$.
The left-hand side of the above equality  represents the  Coulomb energy of the continuous charge distribution 
$\rho(\lambda)$ located in the interval $(0,\lambda)$ in the two-dimensional plane, and 
two  point charges $q_1=q_2=-1/2$ located at the points $\lambda_1=0$ and $\lambda_2=4$.
\section{Conclusions \label{conc}}
We have studied some general  spectral properties of the one-dimensional Sturm-Liouville problem
for the discrete Schr\"odinger equation with the Dirichlet boundary conditions. Both cases of the finite-interval 
and semi-infinite problems were considered. 

For the finite-interval problem with even number of sites $2M$ and an arbitrary even potential, 
it was shown, that its eigenvalues should satisfy the infinite set of polynomial constrains of increasing 
degrees. Though the number of these constrains is infinite, no more than $M$ of them 
are independent. It is simple to find few
initial small-degree polynomials in this set, but explicit calculation of subsequent 
polynomials of higher degrees becomes more and more 
difficult. Nevertheless, we have obtained in the explicit form one polynomial constrain from this set, 
which has the degree $M$, see equation  \eqref{pr0}.
It leads to the effective Coulomb interaction between the eigenvalues, which correspond to even and odd
eigenstates.

The scattering problem for the discrete  one-dimensional 
Schr\"odinger equation  in the half-line has been analysed as the $M\to\infty$ limit of the 
described above $2M$-site discrete Sturm-Liouville problem in the finite-interval.
It was  proved, that the scattering phase 
of the discrete scattering problem \eqref{Sch}-\eqref{Dbc} should satisfy condition \eqref{EQ},
if: (i) the potential has a compact support,  (ii) the spectrum of the Hamiltonian is purely continuous, 
$\sigma[H]=(0,4)$, and (iii) the Jost function takes nonzero values on its end points $\lambda=0$ and
$ \lambda=4$. Constrain \eqref{EQ} admits the electrostatic interpretation \eqref{Coul2}, as its finite-interval
counterpart  \eqref{pr0}.

We did not try to prove \eqref{EQ} for the most general case of the discrete semi-infinite scattering problem. 
It is natural to expect, however, that
it should hold for some more general class of potentials, which  vanish fast enough at infinity, 
though do not have a compact support. On the other hand, in the case of the potentials with discrete spectrum
and/or semi-bound states (the latter appear if the Jost function has zeroes at the end points of the 
continuous spectrum, see \cite{ChaSab}), some modified forms of \eqref{EQ} should also exist.

We believe, that obtained results could be useful for the theory of Anderson localisation and for the
theory of random matrices.
\ack I am thankful to H.~W.~Diehl for fruitful discussions. 
Support of this work by Deutsche Forschungsgemeinschaft (DFG) via  grant Ru 1506/1 is also gratefully acknowledged.
\appendix
\section{Proof of equality \eqref{pr2} \label{free}}

Let us start from the following auxiliary
\begin{lemma} The following equality holds for all natural  $M$ and integer  $n$:
\begin{equation}\label{l2}
\fl\prod_{m=1}^{2M+1}\left\{4 \sin^2\left[\frac{ (2m-1)\pi}{2(2M+1)}-\alpha
\right]- 4 \sin^2\left[\frac{ 2n\pi}{2(2M+1)}\right]  \right\}=4\cos^2[\alpha(2M+1)].
\end{equation}
\end{lemma}
\begin{proof}
Denote
\begin{equation}\label{gdef}
g(\alpha)=\prod_{m=1}^{2M+1}\left\{4 \sin^2\left[\frac{ (2m-1)\pi}{2(2M+1)}-\alpha
\right]- 4 \sin^2\left[\frac{ 2n\pi}{2(2M+1)}\right]  \right\}.
\end{equation}
The symmetry properties of the function $g(\alpha)$ 
\begin{eqnarray*}
g(-\alpha)=g(\alpha), \\
g\left(\alpha +\frac{\pi}{2M+1}\right)=g(\alpha)
\end{eqnarray*}
follow from (\ref{gdef}).

Function $g(\alpha)$ is analytical in the complex $\alpha$-plane and has the second order zeroes at the
points
\begin{equation}\label{zerg}
\alpha_l=\frac{\pi}{2(2M+1)}+\frac{\pi l}{2M+1}, \quad\quad l=0,\pm 1,\pm 2, \ldots
\end{equation}
It follows from (\ref{zerg}) that the
function
\begin{equation}\label{R}
R(\alpha)=\frac{g(\alpha)}{4\cos^2[\alpha(2M+1)]}
\end{equation}
is analytical and has no zeroes in the complex $\alpha$-plane.
Furthermore, this function is rational in the
variable $z=\rme^{\rmi \alpha}$ and approaches to 1 at $z\to \infty$ and at $z\to 0$.
Therefore, $R(\alpha)=1$.
\end{proof}

Putting $\alpha=0$ in (\ref{l2}) we find
\begin{equation}\label{l20}
\prod_{m=1}^{2M+1}\left\{4 \sin^2\left[\frac{ (2m-1)\pi}{2(2M+1)}
\right]- 4 \sin^2\left[\frac{ 2n\pi}{2(2M+1)}\right]  \right\}=4.
\end{equation}
Since
\begin{eqnarray*}
 \prod_{j=1}^{2M+1}\left\{4 \sin^2\left[\frac{ (2m-1)\pi}{2(2M+1)}
\right]- 4 \sin^2\left[\frac{ 2n\pi}{2(2M+1)}\right]  \right\}=\\
\fl\left[\prod_{m=1}^{M}\left\{4 \sin^2\left[\frac{ (2m-1)\pi}{2(2M+1)}
\right]- 4 \sin^2\left[\frac{ 2n\pi}{2(2M+1)}\right]  \right\}\right]^2\,4\left(1- \sin^2\left[\frac{ 2n\pi}{2(2M+1))}\right]\right),
\end{eqnarray*}
we get
\begin{equation*}
\fl\left[\prod_{m=1}^{M}\left\{4 \sin^2\left[\frac{ (2m-1)\pi}{2(2M+1)}
\right]- 4 \sin^2\left[\frac{ 2n\pi}{2(2M+1)}\right]  \right\}\right]^2=\left[\cos\frac{\pi n}{2M+1}\right]^{-2},
\end{equation*}
or
\begin{equation}\label{pro}
\fl\prod_{m=1}^{M}\left\{4 \sin^2\left[\frac{ (2m-1)\pi}{2(2M+1)}
\right]- 4 \sin^2\left[\frac{ 2n\pi}{2(2M+1)}\right]  \right\}=(-1)^n\left[\cos\frac{\pi n}{2M+1}\right]^{-1}.
\end{equation}
To fix the sign of the right-hand side of (\ref{pro}), we have taken into
account that just the first $n$ factors in the product in the left-hand side
are negative at $n=1,\ldots,M$.

Thus,
\begin{eqnarray}\label{2}
\fl\prod_{n=1}^M\prod_{m=1}^M\left[4 \sin^2\frac{ (2m-1)\pi}{2(2M+1)}-
 4 \sin^2\frac{2 n\pi}{2(2M+1)}\right]=
\prod_{n=1}^M\frac{(-1)^{n}}{\cos\frac{\pi n}{2M+1}}.
\end{eqnarray}
To determine the product in the right-hand side we use
formula 1.392.1  in  Ref. \cite{Gradshteyn}:
\begin{equation}\label{GR2}
\sin n x=2^{n-1}\prod_{k=0}^{n-1}\sin\left(x+\frac{\pi k}{n}\right).
\end{equation}
For $n=2M+1$, $x=\pi/2$, we get from (\ref{GR2})
\begin{equation*}
\prod_{k=0}^{2M}\cos\left(\frac{\pi k}{2M+1}\right)=2^{-2M}(-1)^M,
\end{equation*}
providing
\begin{equation}\label{1}
\prod_{k=1}^{M}\cos\left(\frac{\pi k}{2M+1}\right)=2^{-M}.
\end{equation}
Substitution of (\ref{1}) into (\ref{2}) leads finally to (\ref{pr2}).

\end{document}